\long\def\skipit#1{} 

\newcommand{\bu}{\bullet}

\newcommand{\lra}{\longrightarrow}

\newcommand{\mdef}[1]{\textit{\textbf{#1}}}
\newcommand{\noi}{\noindent}

\newcommand{\ov}{\overline}

\newcommand{\vsa}{\vskip-12pt}
\newcommand{\vsb}{\vskip-6pt}

\def\firstterminal{p}
\def\secondterminal{q}

\newcommand{\subsecheadfont}[1]{{\fontsize{13pt}{11pt}\selectfont{#1}}}

\newcommand{\subhead}[1]{\smallskip\begin{center}{\bf \subsecheadfont{#1}}\end{center} \smallskip}

\documentclass[12pt]{amsart}     

\usepackage{amsfonts}
\usepackage{amsmath}
\usepackage{amssymb}
\usepackage{ifthen,calc}
\usepackage{color}
\usepackage{epsfig}
\usepackage{graphicx}
\usepackage{latexsym}
\usepackage{multicol}
\usepackage[format=hang, margin=10pt]{caption}

\DeclareMathOperator{\pgd}{pgd}

\graphicspath{{VA-figs}}

\newcounter{hours}
\newcounter{minutes}
\newcommand{\printtime}{
	\setcounter{hours}{\time/60}%
	\setcounter{minutes}{\time-\value{hours}*60}
	\ifthenelse{\value{hours}<10}{0}{}\thehours:%
	\ifthenelse{\value{minutes}<10}{0}{}\theminutes}

\numberwithin{equation}{section}
\numberwithin{figure}{section}
\numberwithin{table}{section}

\newtheorem{thm}{Theorem}[section]
\newtheorem{cor}[thm]{Corollary}
 
\newtheorem{prop}[thm]{Proposition}

\newtheorem{algo}{Algorithm}[section]

\theoremstyle{definition}

\newtheorem{J-com}{JG-comment}[section]

\begin{document}

\title[Genus Distributions of cubic series-parallel graphs]
{\large{Genus Distributions \\[4pt] of Cubic Series-Parallel Graphs}}

\author{Jonathan L. Gross}  
\address{Department of Computer Science\\
Columbia University, New York, NY 10027, USA}  
\email{gross@cs.columbia.edu}
\urladdr{http://www.cs.columbia.edu/~gross/}

\author{Michal Kotrb\v{c}\'ik}
\address{Department of Computer Science\\
Masaryk University, Brno 602 00, Czech Republic}
\email{kotrbcik@fi.muni.cz}
\urladdr{http://www.fi.muni.cz/~qkotrbc/}

\author{Timothy Sun}
\address{Department of Computer Science\\
Columbia University, New York, NY 10027, USA}
\email{tim@cs.columbia.edu}
\urladdr{http://www.cs.columbia.edu/~tim/}

\thanks{}

\date{}

\begin{abstract}  
We derive a quadratic-time algorithm for the genus distribution of any 3-regular, biconnected  series-parallel graph, which we extend to any biconnected series-parallel graph of maximum degree at \hbox{most 3}.  Since the biconnected components of every graph of treewidth 2 are series-parallel graphs, this yields, by use of bar-amalgamation, a quadratic-time algorithm for every graph of treewidth at most 2 and maximum degree at most 3.  

\noi \textbf{Version: \printtime\quad\today\quad} 
\end{abstract} 

\maketitle                   

\section{\large{Introduction}}  

For $i = 0,1,2, \ldots$, let $g_i(G)$ be the number of topologically distinct cellular embeddings of the graph $G$ in the orientable surface $S_i$ of \hbox{genus $i$}.  The \mdef{genus distribution} of the graph $G$ is the sequence of numbers \par\vsa\vsb
\begin{equation} \label{eq:gd}
g_i(G) : i=0,1, \ldots
\end{equation}
By the interpolation principle (see Theorem 3.4.1 of \cite{GrTu87} or Theorem  4.5.3 of \cite{MoTh01}), the set $\{i:g_i(G)>0\}$ is a set of consecutive integers.  The smallest number in this set is the \mdef{minimum genus} of the graph $G$, and the largest is the \mdef{maximum genus} of $G$.

\enlargethispage{-12pt}

The main focus of this paper is the derivation of a quadratic-time algorithm for the genus distribution of any 3-regular, biconnected \textit{series-parallel graph}.   This algorithm is readily extended to a quadratic-time algorithm for the genus distribution of any graph of treewidth at most~2 and maximum degree at most 3.  The simplicity with which this specialized algorithm can be implemented, or applied by hand with the aid of a spreadsheet, distinguishes it from the recently derived quadratic-time algorithm \cite{Gr12} for the genus distribution of any class of graphs of fixed treewidth and bounded degree.  
\medskip


\subhead{Basic results on genus distribution}  

Five fundamental papers \cite{GKP10, Gr11a, PKG10, KPG10, PKG12} of the first author and his co-authors Khan and Poshni have established methods for calculating the genus distribution of a graph that is constructed by various kinds of amalgamation of graphs of known genus distribution.  These methods involve \textit{partitioned genus distributions} and \textit{productions}.  In order to develop an algorithm for  a specific class of graphs, the starting point is to formulate a recursive specification of the graphs in that class, in which the operations used to create larger graphs from smaller graphs are varieties of amalgamation or self-amalgamation.  Then methods similar to those of the five fundamental papers are used to calculate the genus distribution recursively.  This paradigm was used successfully in calculating the genus distributions of 3-regular outerplanar graphs \cite{Gr11b}, of 4-regular outerplanar graphs \cite{PKG11}, of Halin graphs \cite{Gr11c}, and of the $3\times n$-mesh graphs \cite{KPG12}.  We adopt the same paradigm in this paper. 
\medskip

\enlargethispage{12pt}

\subhead{Connections of treewidth to embedding problems}  

Since the introduction of the concept of \textit{treewidth} by Robertson and Seymour, bounding the treewidth has been widely used to obtain polynomial-time algorithms for problems that are otherwise NP-hard.  In particular, deciding whether an arbitrarily selected graph can be embedded in a given surface is NP-complete \cite{Th89};   however, for any class of graphs of bounded treewidth, Kawarabayashi, Mohar, and Reed \cite{KMR08} have derived a linear-time algorithm for calculating the minimum genus.   

Although outerplanar graphs have {treewidth} 2,  and although Halin graphs and $P_3\times P_n$ meshes have {treewidth}  3 (see \cite{Bo98}), decomposition trees have not occurred in the calculation of specific genus distributions in any papers as yet.   Nonetheless, low treewidth plays an implicit role in the recursive specification of the family of graphs in each of those papers.  Similarly, in the present paper, low treewidth plays an implicit role, since it allows for a simple recursive construction of the graphs under consideration.   
\medskip

\eject
\subhead{Terminology}  

In what follows, a \mdef{graph} is taken to be connected and devoid of self-loops, unless something else can be inferred from the immediate context.   Multi-edges are to be expected. We use $V_G$ and $E_G$ to denote the vertex set and edge set of a graph $G$.  A connected graph is \mdef{biconnected} if it has no cutpoints.  
\smallskip

The \mdef{embeddings} in this paper are cellular embeddings in oriented surfaces.  The terminology used here is predominantly consistent with \cite{GrTu87} and \cite{BWGT09}.  See also \cite{MoTh01}, for a slightly different approach.   We abbreviate ``face-boundary walk'' as \mdef{fb-walk}. 
\smallskip

A \mdef{two-terminal series-parallel graph}  is a doubly vertex-rooted graph $(G,\firstterminal,\secondterminal)$,  as per the following recursive definition.
\begin{description}
\item[$B$] The graph $(K_2,\firstterminal,\secondterminal)$ is a two-terminal series-parallel graph, where $\firstterminal$ and $\secondterminal$ are the vertices of $K_2$, called the \mdef{source root} and the \mdef{target root}, respectively. 
\item[$R_1$] \mdef{series operation} $(G,\firstterminal,\secondterminal)\odot_s(G',\firstterminal',\secondterminal')$ Target root $\secondterminal$ of $G$ is merged with source root $\firstterminal'$ of $G'$.  The amalgamated graph $G\odot_sG'$ with roots $\firstterminal$ and $\secondterminal'$, as in Figure \ref{fig:SeriesPar}(a), is a two-terminal series-parallel graph.    
\item[$R_2$] \mdef{parallel operation} $(G,\firstterminal,\secondterminal)\odot_p(G',\firstterminal',\secondterminal')$  The result of merging source root $\firstterminal$ with source root $\firstterminal'$ , and also merging target root $\secondterminal$ with target root $\secondterminal'$, as in Figure \ref{fig:SeriesPar}(b), is a two-terminal series-parallel graph.    
\end{description}

\begin{figure} [ht]  
\centering 
    \includegraphics[width=4.2in]{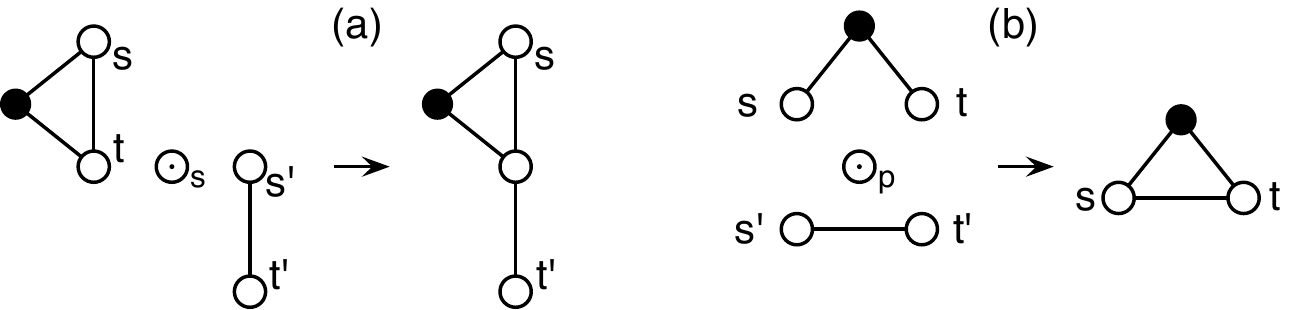}  \vskip-6pt
\caption{Operations on series-parallel graphs.}
\label{fig:SeriesPar}
\end{figure}  
\noi A graph $G$ is a \mdef{series-parallel graph} if there is a choice of terminals $\firstterminal$ and $\secondterminal$ such that $(G,\firstterminal,\secondterminal)$ is a two-terminal series-parallel graph.  Our definition here is consistent with that of \cite{Bo98} and \cite{Ep92}.  A third operation, called a \textit{jackknife operation} is allowed by \cite{BPT09}, and the resulting class of graphs that they call ``series-parallel'' is equivalent to that of \cite{Du65} --- see Remark 8.1 of \cite{BPT09}.   The ``series-parallel graphs" in \cite{Du65} are identified there as the graphs with no embedded ``Wheatstone bridge" (which is Duffin's terminology for a $K_4$ topological minor).  These are precisely the graphs that have no $K_4$-minor (e.g., see Proposition 1.7.2 of \cite{Di06}).  According to Theorem 17 of \cite{Bo98}, the graphs without a $K_4$-minor are exactly the graphs of treewidth at most 2.  It follows that our extended genus-distribution algorithm can be applied to any of them in quadratic time, and that we do not need to further explore the distinctions between the varying definitions of ``series-parallel graphs''. 
\medskip

\subhead{Outline of this paper} 

Section \ref{sec:3regSP} derives a characterization of 3-regular, biconnected series-parallel graphs that facilitates the genus distribution algorithm for that family of graphs.  Section \ref{sec:prods} introduces the concepts of \textit{partitioned genus distributions} and \textit{productions}. The top-level description of an algorithm for the genus distribution of any 3-regular, biconnected series-parallel graph is given in Section \ref{sec:algo}.   \hbox{Sections \ref{sec:dmt} and  \ref{sec:3to1}} derive the productions needed to complete the calculation, as well as their application to calculating the genus distribution of an illustrative example.  Section \ref{sec:3to1} also gives proof that the algorithm runs in quadratic time.  \hbox{Section \ref{sec:TW2deg3}} extends the algorithm to all graphs of treewidth 2 and maximum \hbox{degree 3}.

This paper is almost entirely self-contained, except for some details of the well-established concept of partitioned genus distributions and of the methods  (as in \cite{GKP10} and \cite{PKG10}) for constructing productions (which are quite necessary for the algorithm).  Prior experience with calculating genus distributions of graph amalgamations, especially as in \cite{Gr11b} and \cite{Gr11c}, is likely to be quite helpful.   

\bigskip

\enlargethispage{12pt}

\section{\large{Cubic Biconnected Series-Parallel Graphs} \label{sec:3regSP}} 

The \mdef{dipole} $D_n$ is the graph with two vertices and an $n$-fold multi-edge joining them.  In this section, we prove that every 3-regular, biconnected series-parallel graph can be obtained by iterated application of the following operation to the dipole $D_3$. 
\begin{description}
\item[$\tau$]  Trisect an arbitrary edge $e$ of a graph $G$ and install a new edge in parallel to the ``middle third" of edge $e$.
\end{description}
This operation, which is applicable to an non-empty graph, is called a \mdef{dmt-step} (``dmt'' is an abbreviation of ``double the middle third''), is illustrated by Figure \ref{fig:dmt}.  

\begin{figure} [ht]  
\centering 
    \includegraphics[width=3.3in]{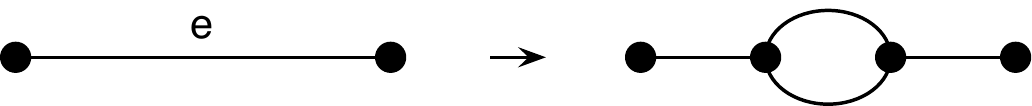}  \vskip-6pt
\caption{A dmt-step: double the middle third.}
\label{fig:dmt}
\end{figure}  

\noi We define a second operation, which is applicable to 3-regular multigraphs, other than the dipole $D_3$. 
\begin{description}
\item [$\tau^{-1}$]   In a 3-regular graph $G$, let $d$ and $e$ be edges that share the same two endpoints, $u$ and $v$, and in which no other edge shares these two endpoints.  Delete edge $d$, and then smooth away vertices $u$ and $v$.
\end{description}
We observe that the operation $\tau^{-1}$ can be used as an inverse to the operation $\tau$.  \medskip

The proof of our characterization of 3-regular, biconnected series-parallel graphs uses the following four propositions. 
\smallskip

\begin{prop} 
A graph $G$ has treewidth at most 2 if and only if every biconnected component of $G$ is a series-parallel graph.  
\label{prop:TW2-SP} 
\end{prop}
\begin{proof} \vsb
This is Theorem 42 of \cite{Bo98}. 
\end{proof}
\smallskip

\begin{prop}
A graph $G$ has treewidth at most 2 if and only if contains no $K_4$-minor.   
\label{prop:zip} 
\end{prop}
\begin{proof} \vsb
This follows immediately from Theorem 17 of \cite{Bo98}. 
\end{proof}
\smallskip

\begin{prop}  \label{prop:inherit} 
Let $G$ be a 3-regular, biconnected series-parallel graph, and let $G'$ be a graph with at least two vertices, obtained by applying operation $\tau^{-1}$ to edges $d$ and $e$  \hbox{of $G$} , with shared endpoints $u$ and $v$.  Then $G'$ is 3-regular, biconnected, and series-parallel. 
\end{prop}
\begin{proof}\vsb
Since $G$ is 3-regular, all of its vertices are 3-valent.  The operation $\tau^{-1}$ eliminates two vertices of $G$ without changing the valences of the remaining vertices.  Thus, the graph $G'$ is 3-regular.  

Let $x$ and $y$ be any two vertices of $G'$.  Since $G$ is biconnected, there is a pair of internally disjoint paths in $G$ joining $x$ and $y$, by Menger's theorem.   If one of these paths contains either of the vertices $u$ or $v$, then it also contains one of the edges $d$ or $e$, since $G$ is 3-regular, and the other path contains neither vertex $u$ nor vertex $v$.  Thus the images of the two internally disjoint path in $G$ are two internally disjoint paths in $G'$. Hence, the graph $G'$ is biconnected.

Since the graph $G$ is series-parallel, its treewidth is at most 2, by Proposition \ref{prop:TW2-SP}.  It follows from Proposition \ref{prop:zip} that $G$ has no $K_4$-minor.  Accordingly, the graph $G'$ has no $K_4$ minor.  Therefore, by Proposition \ref{prop:zip}, the graph $G'$ has treewidth at most 2.  We conclude from Proposition \ref{prop:TW2-SP} that the graph $G'$ is series-parallel. 
\end{proof}
\smallskip

\begin{prop}[Dirac's theorem] \label{prop:Dirac} 
Let $G$ be a biconnected simple graph of minimum degree 3.  Then $G$ contains a subgraph that is homeomorphic to the complete graph $K_4$. 
\end{prop}
\begin{proof}\vsb
This is a theorem of \cite{Di52}.
\end{proof}
\smallskip

\begin{thm}
Let $G$ be a biconnected 3-regular series-parallel graph. Then there exist vertices $\firstterminal,\secondterminal\in V_G$ such that the two-terminal series-parallel graph $(G, \firstterminal, \secondterminal)$ is derivable from $(D_3, \firstterminal, \secondterminal)$ by a sequence of applications of the operation $\tau$. 
\end{thm}
\begin{proof} \vsb
Let $H$ be a smallest graph obtainable by iterative application of operation $\tau^{-1}$ to the graph $G$, that is, a graph such that no pair of vertices is joined by exactly two edges.  By Proposition \ref{prop:inherit}, the graph $H$ is 3-regular, biconnected, and series-parallel. We observe that $H$ cannot be simple, lest it contain, by Dirac's theorem, a homeomorphic copy \hbox{of $K_4$}, a contradiction, in view of Propositions \ref{prop:zip} and \ref{prop:TW2-SP}.  Accordingly, there is a pair of vertices $\firstterminal,\secondterminal\in V_H$ with at least two edges joining them.  Since there cannot be exactly two edges joining $\firstterminal$ and $\secondterminal$, by the minimality of the graph $H$, and since $H$ is 3-regular, it follows that $H\cong D_3$.  Reversing the sequence of $\tau^{-1}$-operations, we obtain a derivation of $(G,\firstterminal,\secondterminal)$ from $(D_3,\firstterminal,\secondterminal)$ by iterative application of the operation $\tau$. 
\end{proof}
\smallskip

We define a \mdef{dmt-string} to be a graph obtained by iterative application of  dmt-steps to the graph $K_2$.  We observe that each dmt-string has two univalent vertices and that all other vertices are trivalent. 
\medskip

\begin{cor}
Let $(G,\firstterminal,\secondterminal)$ be a 3-regular, biconnected two-terminal series-parallel graph.  Then $(G,\firstterminal,\secondterminal)$ can be represented by a set of three dmt-strings, each with a univalent $\firstterminal$-vertex and a univalent $\secondterminal$-vertex, from which $(G,\firstterminal,\secondterminal)$ is formed by two parallel operations. 
\end{cor}
\bigskip

\enlargethispage{12pt}

\section{\large{Partials and Productions} \label{sec:prods}}  

When calculating the genus distribution of a family of graphs, we commonly use a finer partition of the embeddings during the intermediate steps.  For computational purposes, we need to isolate subsets of embeddings upon which the surgical operations used in the recursive construction of that family have the same effect.  Whereas the genus distribution of a graph is an inventory according only (as per \eqref{eq:gd}) to the genus of the embedding surface, a  \mdef{partitioned genus distribution} refines the genus distribution of a rooted graph, according to the incidence of fb-walks on the roots, which is the critical factor in the behavior of a surgical operation on an embedding. 
In other words, a partitioned genus distribution is a partition of all embeddings of the graph with a given genus into into several types, which allows us to keep under control the structure of the faces incident with the roots before, respectively after the amalgamation.
The cells of the finer partition are called \mdef{partials}.  In this context, we sometimes abbreviate genus distribution as \mdef{gd} and partitioned genus distribution as \mdef{pgd}. 
\smallskip

Calculating the genus distributions of 3-regular series-parallel graphs involves amalgamating subgraphs at pairs of terminals, such that the sum of the degrees of an amalgamated pair  of vertices is at most 3.  Thus, the possible degrees of the terminals prior to the final operation are 1 and 2.  When both terminals are univalent, the genus distribution of $(G,\firstterminal,\secondterminal)$ is partitioned into the following partials: 
\begin{eqnarray*}
uu_i^\bu(G,\firstterminal,\secondterminal) &=& \hbox{the number of embeddings $G\to S_i$ such that} \\ [-3pt]
&& \hbox{terminals $\firstterminal$ and $\secondterminal$ do not occur on the same fb-walk;} \\
uu_i'(G,\firstterminal,\secondterminal) &=& \hbox{the number of embeddings $G\to S_i$ such that} \\ [-3pt]
&& \hbox{terminals $\firstterminal$ and $\secondterminal$ occur on the same fb-walk.}  
\end{eqnarray*}
The letter $u$ in the name of the partial is a mnemonic for \textit{univalent}.  For every $i=0, 1, 2, \ldots$, the set of all embeddings of $(G,\firstterminal,\secondterminal)$ with genus $i$ gives a partitioned genus distribution given by the formula
$$g_i(G,\firstterminal,\secondterminal) ~=~ uu_i^\bu(G,\firstterminal,\secondterminal) + uu_i'(G,\firstterminal,\secondterminal)$$

When terminal $\firstterminal$ is univalent and terminal $\secondterminal$ is bivalent, the letters $d$ or $s$ in the name of the partial mean, respectively, that $\secondterminal$ occurs on two \textit{different} fb-walks or that $\secondterminal$ occurs twice on the \textit{same} fb-walk. There are four partials:
\begin{eqnarray*}
ud_i^\bu(G,\firstterminal,\secondterminal) &=& \hbox{the number of embeddings $G\to S_i$ such that} \\ [-3pt]
&& \hbox{terminal $\firstterminal$ occurs on neither fb-walk incident at $\secondterminal$;} \\
ud_i'(G,\firstterminal,\secondterminal) &=& \hbox{the number of embeddings $G\to S_i$ such that} \\ [-3pt]
&& \hbox{terminal $\firstterminal$ occurs on one fb-walk incident at $\secondterminal$;}  \\
us_i^\bu(G,\firstterminal,\secondterminal) &=& \hbox{the number of embeddings $G\to S_i$ such that} \\ [-3pt]
&& \hbox{terminal $\firstterminal$ does not occur on the fb-walk incident at $\secondterminal$;} \\
us_i'(G,\firstterminal,\secondterminal) &=& \hbox{the number of embeddings $G\to S_i$ such that} \\ [-3pt]
&& \hbox{terminal $\firstterminal$ occurs on the fb-walk incident at $\secondterminal$.}   
\end{eqnarray*}

As before,  for every $i=0, 1, 2, \ldots$, the set of all embeddings of $(G,\firstterminal,\secondterminal)$ with genus $i$ has a partitioned genus distribution according to the formula
$$g_i(G,\firstterminal,\secondterminal) ~=~ ud_i^\bu(G,\firstterminal,\secondterminal) + ud_i'(G,\firstterminal,\secondterminal)+ us_i^\bu(G,\firstterminal,\secondterminal) + us_i'(G,\firstterminal,\secondterminal)$$


\noi Similarly, when $\firstterminal$ is bivalent and $\secondterminal$ is univalent, there are four partials:
\begin{eqnarray*}
du_i^\bu(G,\firstterminal,\secondterminal) &=& \hbox{the number of embeddings $G\to S_i$ such that} \\ [-3pt]
&& \hbox{terminal $\secondterminal$ occurs on neither fb-walk incident at $\firstterminal$;} \\
du_i'(G,\firstterminal,\secondterminal) &=& \hbox{the number of embeddings $G\to S_i$ such that} \\ [-3pt]
&& \hbox{terminal $\secondterminal$ occurs on one fb-walk incident at $\firstterminal$;}  \\ 
su_i^\bu(G,\firstterminal,\secondterminal) &=& \hbox{the number of embeddings $G\to S_i$ such that} \\ [-3pt]
&& \hbox{terminal $\secondterminal$ does not occur on the fb-walk incident at $\firstterminal$;} \\
su_i'(G,\firstterminal,\secondterminal) &=& \hbox{the number of embeddings $G\to S_i$ such that} \\ [-3pt]
&& \hbox{terminal $\secondterminal$ occurs on the fb-walk incident at $\firstterminal$.}  
\end{eqnarray*}
\smallskip

\enlargethispage{12pt}

Suppose that $p^1, p^2, \ldots, p^s$ is a set of partials for a genus distribution.  A \mdef{production} for a given surgical operation that transforms a graph embedding $X\to S_i$ (or a tuple of graph embeddings) into a set of graph embeddings of the graph $Y$ is an algebraic rule of this form:  
\begin{equation}
p^j_i(X) \lra c_1 p^1_{f^j_1(i)}(Y) + \cdots + c_t p^s_{f^j_s(i)}(Y)
\end{equation}
The left side is called the \mdef{antecedent}, and the right side is called the \mdef{consequent}. The meaning is that the operation transforms a single embedding of graph $X$ of type $p^j$ on the orientable surface $S_i$ of genus~$i$ into a set of embeddings of the graph  $Y$, of which $c_k$ are of type $p^k$  on the surface $S_{f^j_k(i)}$, for each $i$, $j$, and $k$.  A drawing is usually used as an aid in deriving the production and in proving its correctness. The names of the graphs and their roots can be suppressed when there is in context no ambiguity.  Thus, we may write 
\begin{equation}
p^j_i \lra c_1 p^1_{f^j_1(i)} + \cdots + c_t p^s_{f^j_s(i)}
\end{equation}

In general, when there are $n$ partials, a surgical operation on two graphs is represented by $n^2$ productions for the partials.  It is clear to someone familiar with the use of partials and productions that it is possible to represent the parallel operation and the series operation by respective lists of productions, one for each ordered pair of partials. This would also lead to an algorithm that requires quadratic-time.  Since our present objective is an algorithm that can be described concisely and calculated by hand for small graphs, we intend here to construct shorter lists of productions.  
\bigskip

\enlargethispage{18pt}

\section{\large{An Algorithm} \label{sec:algo}} 

Our Algorithm \ref{algo:biconn} for calculating the genus distribution of any cubic, biconnected, two-terminal series-parallel graph $G$ has five steps.   
\vskip0.5cm

\noindent
\fbox{\parbox{4.9in}{  
\begin{algo}
\textbf{Genus distribution algorithm for a cubic, biconnected, series-parallel graph $G$.}   \label{algo:biconn}
\end{algo}
\noi Input: A 3-regular biconnected series-parallel graph $G$. \par
\noi Output:  The genus distribution of the graph $G$.  \smallskip
\begin{enumerate}
\item Choose the endpoints $\firstterminal$ and $\secondterminal$ of an edge as the terminals. \label{algo:Epp}
\item Determine the three dmt-strings $N^1, N^2, N^3$ corresponding to the graph $(G,\firstterminal,\secondterminal)$. \label{algo:dmt-strings}
\item Calculate the pgd of each of the three dmt-strings $N^1, N^2, N^3$.\label{algo:pgd-dmt}
\item Calculate the pgd of the graph $N^1\odot_p N^2$. \label{algo:pgd-2dmt}
\item Calculate the gd of the graph $G\,=\,(N^1\odot_p N^2)\odot_p N^3$. \label{algo:pgd-3dmt}
\end{enumerate}
}} \vskip0.5cm

\enlargethispage{18pt}

\noi Step \eqref{algo:Epp}.  Lemma 9 of \cite{Ep92} shows that a biconnected series-parallel graph is two-terminal series-parallel, for any pair of terminals $\firstterminal$ and $\secondterminal$ that are joined by an edge. Therefore, we can select as terminals $\firstterminal$ and $\secondterminal$ the endpoints of any edge of $G$. 

\noi Step \eqref{algo:dmt-strings}.  Determine the three dmt-strings $(N^1,\firstterminal,\secondterminal), (N^2,\firstterminal,\secondterminal), (N^3,\firstterminal,\secondterminal)$ for the \hbox{graph $G$}, by splitting both of the vertices $\firstterminal$ and $\secondterminal$ of the dipole $D_3$ into three vertices, each an endpoint of one of the edges incident on the split vertex. 

\noi Step \eqref{algo:pgd-dmt}.    It simplifies this calculation if we define a small modification of the parallel operation $\odot_p$.  When we combine two dmt-strings with a parallel operation, we obtain two 2-valent vertices.  Our modified operation $\ov\odot_p$ attaches a spike at each of these 2-valent vertices, as illustrated in Figure \ref{fig:odotp}, so that we once again have a dmt-string.   

\begin{figure} [ht]  
\centering 
    \includegraphics[width=3.5in]{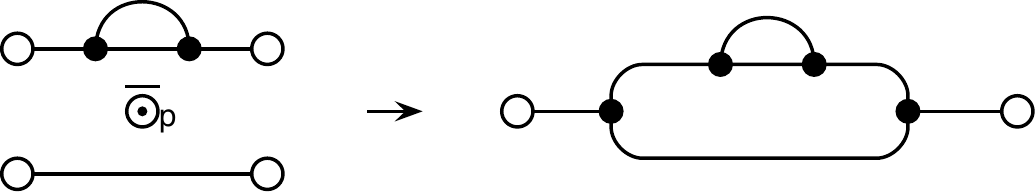}  \vskip-6pt
\caption{The modified parallel operation $\ov\odot_p$.}
\label{fig:odotp}
\end{figure}  

\noi Similarly, our modified operation $\ov\odot_s$ merges the second terminal of the first graph with the first terminal of the second graph, and then smooths away the merged vertex,  as illustrated in Figure \ref{fig:odots}.    

\eject

\begin{figure} [ht]  
\centering 
    \includegraphics[width=3.2in]{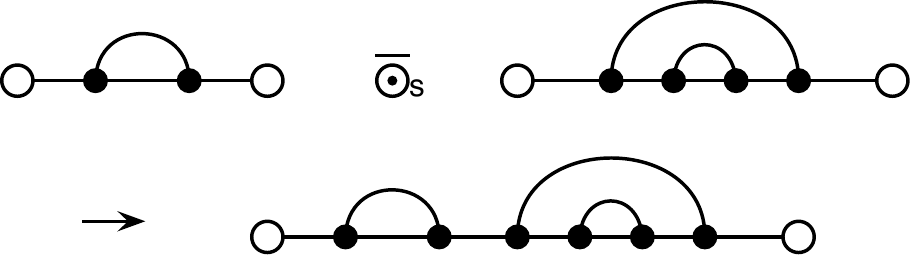}  \vskip-6pt
\caption{The modified parallel operation $\ov\odot_s$.}
\label{fig:odots}
\end{figure}  

\noi The method for calculating the pgd of a dmt-string is given in \hbox{Section \ref{sec:dmt}}.

\enlargethispage{12pt}

\noi Step \eqref{algo:pgd-2dmt}.    This step concerns the composition of an amalgamation of two double-rooted graphs with univalent roots at one pair of univalent roots, followed by a self-amalgamation at the other two roots. See  \hbox{Section \ref{sec:3to1}} for a description of this computation.  

\noi Step \eqref{algo:pgd-3dmt}.   This step is the composition of an amalgamation of two double-rooted graphs with univalent roots at one pair of univalent roots, followed by a self-amalgamation at the other two roots.  See  \hbox{Section \ref{sec:3to1}} for a description of this computation.  

\medskip


\section{\large{PGD of DMT-Strings} \label{sec:dmt}}  

Four parallel productions and four series productions will be sufficient to calculate the values of the partials of any dmt-string.  The following two sets of four productions each are sufficient to calculate all the partials of a dmt-string.  The parallel productions are derived with the aid of Figures \ref{fig:uu*uu}, \ref{fig:uu*uu'}, and \ref{fig:uu'*uu'}.  The series productions are self-evident.  The genus of each resultant embedding is calculated from its Euler characteristic. 

\begin{eqnarray}
\noalign{\vskip-12pt}
\qquad uu_i^\bu\,\ov\odot_p\, uu_j^\bu &\lra& 4uu_{i+j+1}^\bu  \quad\hbox{(Figure \ref{fig:uu*uu})}\label{prod:uu*uu}  \\ 
uu_i^\bu\,\ov\odot_p\, uu_j'  &\lra&  4uu_{i+j+1}' \quad\hbox{(Figure \ref{fig:uu*uu'})} \label{prod:uu*uu'}  \\ 
uu_i'\,\ov\odot_p\, uu_j^\bu &\lra&  4uu_{i+j+1}' \quad\hbox{(mirror of Figure \ref{fig:uu*uu'})}\label{prod:uu'*uu}  \\ 
uu_i'\,\ov\odot_p\, uu_j'  &\lra&  2uu_{i+j}^\bu+ 2uu_{i+j}' \quad\hbox{(Figure \ref{fig:uu'*uu'})} \label{prod:uu'*uu'}  
\end{eqnarray}

\begin{figure} [h!]  
\centering  \vsb
    \includegraphics[width=4.5in]{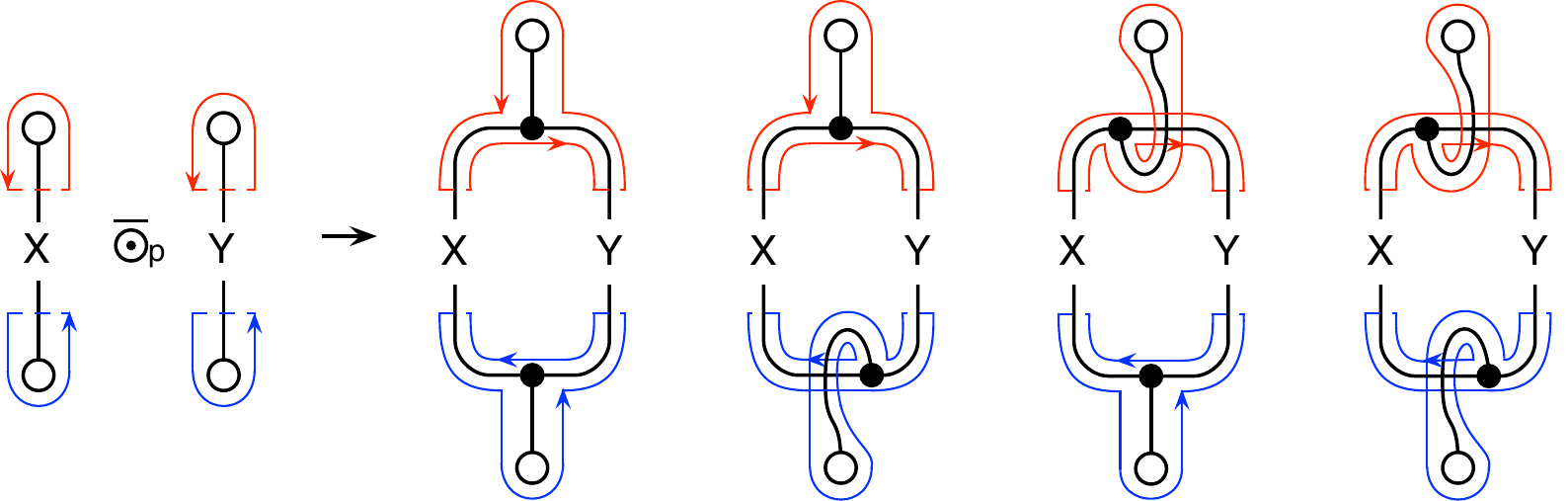}  \vskip-6pt
\caption{Production: $uu_i^\bu\,\ov\odot_p\, uu_j^\bu \lra 4uu_{i+j+1}^\bu$.}
\label{fig:uu*uu}
\end{figure}  

\eject

\begin{figure} [h!]  
\centering 
    \includegraphics[width=4.7in]{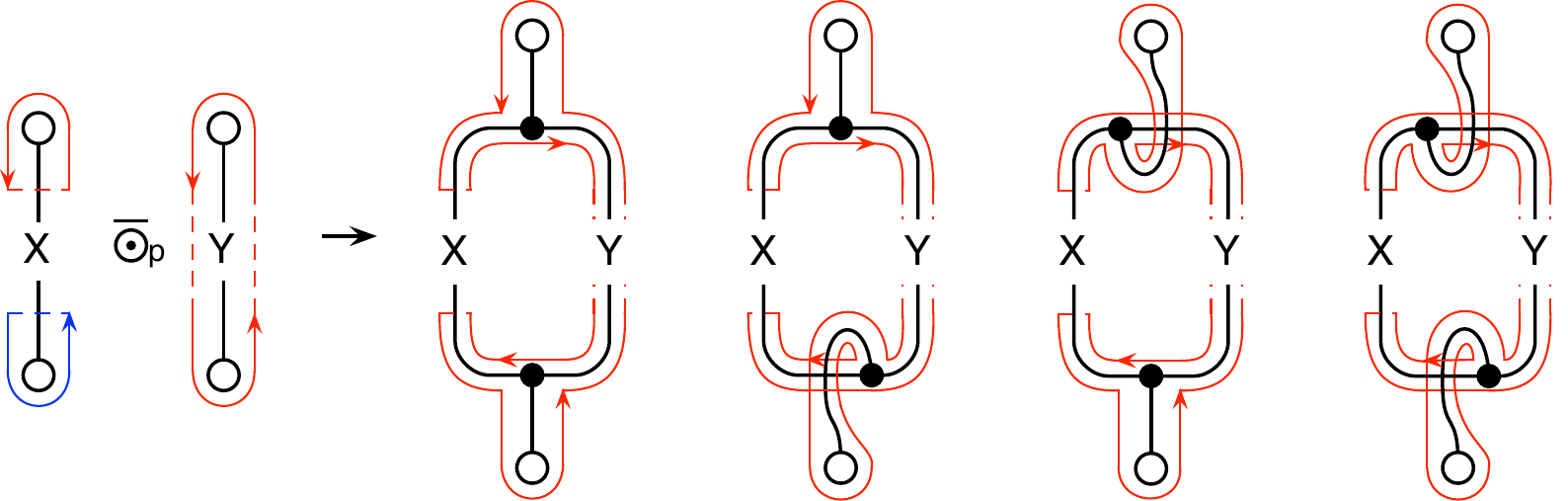}  \vskip-6pt
\caption{Production: $uu_i^\bu\,\ov\odot_p\, uu_j' \lra 4uu_{i+j+1}'$.}
\label{fig:uu*uu'}
\end{figure}  \vskip-12pt

\begin{figure} [h!]  
\centering 
    \includegraphics[width=4.7in]{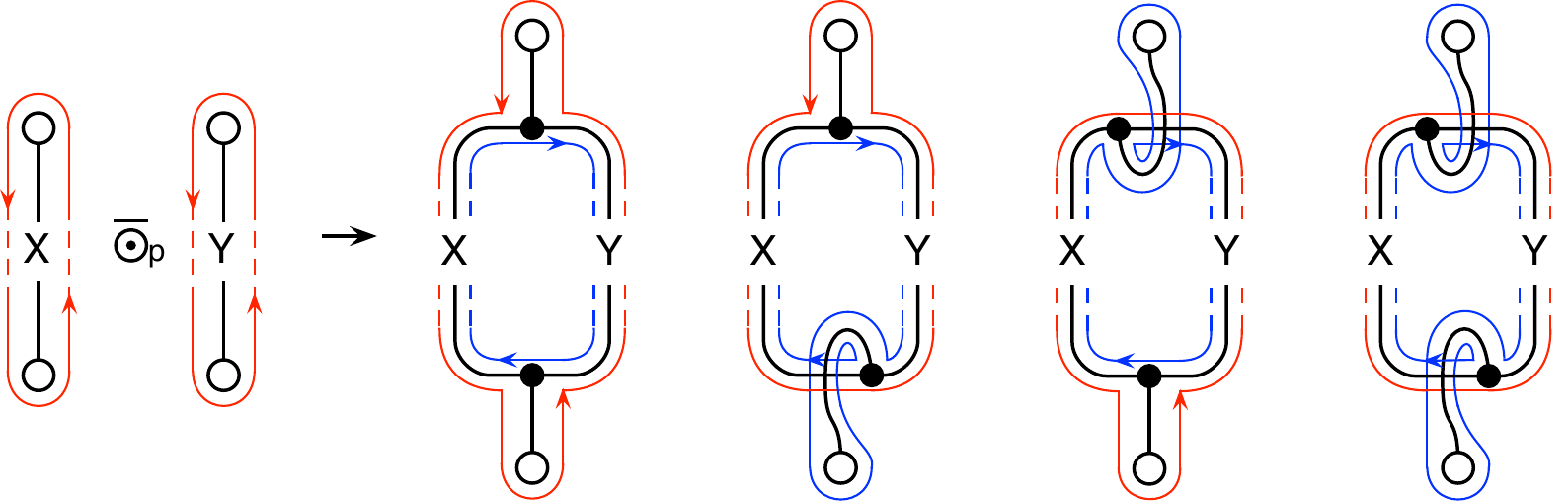}  \vskip-6pt
\caption{Production: $uu_i'\,\ov\odot_p\, uu_j' \lra 2uu_{i+j}' + 2uu_{i+j}^\bu$.}
\label{fig:uu'*uu'}
\end{figure}  
\smallskip

\begin{eqnarray} 
\noalign{\vskip-18pt}
uu_i^\bu\,\ov\odot_s\, uu_j^\bu &\lra& uu_{i+j}^\bu \label{prod:uu-suu}  \\
uu_i^\bu\,\ov\odot_s\, uu_j'  &\lra&  uu_{i+j}^\bu \label{prod:uu-suu'}  \\
uu_i'\,\ov\odot_s\, uu_j^\bu &\lra&  uu_{i+j}^\bu \label{prod:uu'-suu}   \\
uu_i'\,\ov\odot_s\, uu_j'  &\lra&  uu_{i+j}'  \label{prod:uu'-suu'}  
\end{eqnarray}

\medskip

\subhead{Examples} 

Our first example is the dmt-string $\hat D_2$ of Figure \ref{fig:hat-D2}, which is of fundamental use in our further calculations.  Here we use $\pgd$ (which stands for \textit{partitioned genus distribution}) as a function.
\smallskip

\begin{figure} [h!]  
\centering 
    \includegraphics[width=1.5in]{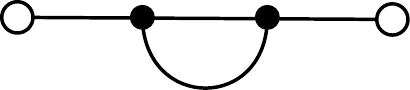}  \vskip-6pt
\caption{The dmt-string $\hat D_2$.}
\label{fig:hat-D2}
\end{figure}  

\noi We observe that $\hat D_2$ is representable as $K_2\,\ov\odot_p\, K_2$. Therefore,
\begin{eqnarray} 
\pgd (\hat D_2) &=&  \pgd (K_2 \,\ov\odot_p\, K_2) \notag \\
&=& uu_0'\,\ov\odot_p\,uu_0' \notag\\
\pgd (\hat D_2) &=& 2uu_0^\bu+2uu_0'  \qquad\hbox{by Prod.\,\eqref{prod:uu'*uu'}} \label{pgd:hat-D2}
\end{eqnarray}

\eject 

We now apply Eq.\,\eqref{prod:uu'*uu'} and the productions above to calculate the partials of the three dmt-strings of the graph of Figure \ref{fig:runex}.

\begin{figure} [h!]  
\centering 
    \includegraphics[width=3in]{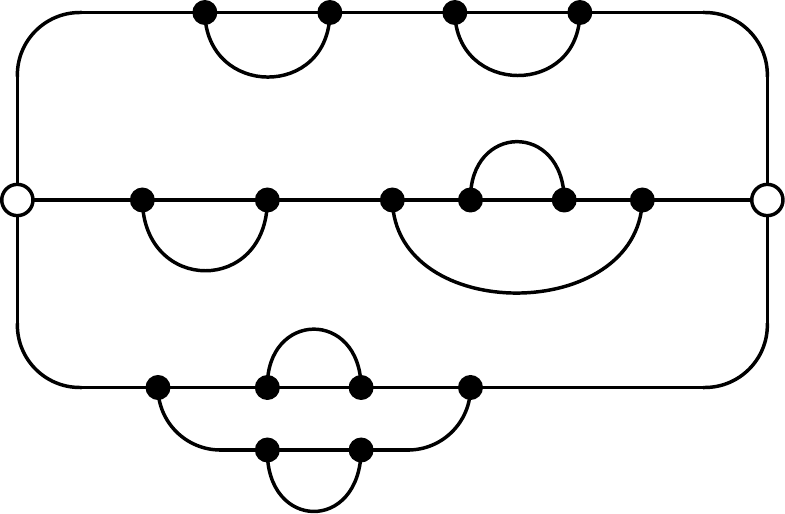}  \vskip-6pt
\caption{A cubic series-parallel graph.}
\label{fig:runex}
\end{figure}  

\noi The top dmt-string $N^1$ is representable as  $\hat D_2\,\ov\odot_s\, \hat D_2$. Therefore, 
\begin{eqnarray}   
\pgd (N^1) &=& \pgd(\hat D_2\,\ov\odot_s\, \hat D_2) \notag \\
&=& (2uu_0^\bu+2uu_0')\,\ov\odot_s\,(2uu_0^\bu+2uu_0') \notag\\
&=& 2uu_0^\bu\,\ov\odot_s\,2uu_0^\bu  + 2uu_0^\bu\,\ov\odot_s\,2uu_0'   \notag \\ 
&& \quad\qquad +\; 2uu_0'\,\ov\odot_s\,2uu_0^\bu  +   2uu_0'\,\ov\odot_s\,2uu_0'   \notag\\
&=& 4uu_0^\bu + 4uu_0^\bu +\; 4uu_0^\bu + 4uu_0' \notag\\
&& \qquad \hbox{by Prods.\,\eqref{prod:uu-suu}, \eqref{prod:uu-suu'}, \eqref{prod:uu'-suu}, \eqref{prod:uu'-suu'}}\notag \\
\pgd (N^1) &=& 12uu_0^\bu+4uu_0'  \label{pgd:N1}
\end{eqnarray}

The dmt-string $N^2$ is representable as $\hat D_2 \,\ov\odot_s\,(\hat D_2 \,\ov\odot_p\, K_2)$.  Therefore,  
\begin{eqnarray}  
\qquad \pgd (N^2)  &=&\pgd(\hat D_2 \,\ov\odot_s\,(\hat D_2 \,\ov\odot_p\, K_2)) \notag \\
&=& (2uu_0^\bu+2uu_0')\,\ov\odot_s\,((2uu_0^\bu+2uu_0')\,\ov\odot_p\, uu_0') \notag \\
&=& (2uu_0^\bu+2uu_0')\,\ov\odot_s\,(8uu_1' + 4uu_0^\bu + 4uu_0') \notag \\
&&\qquad \hbox{by Prods.\,\eqref{prod:uu*uu'}, \eqref{prod:uu'*uu'}}\notag\\
&=& 16uu_1^\bu + 8uu_0^\bu + 8uu_0^\bu + 16uu_1' + 8uu_0^\bu + 8uu_0' \notag \\
&& \qquad \hbox{by Prods.\,\eqref{prod:uu*uu}, \eqref{prod:uu*uu'}, \eqref{prod:uu'*uu}, \eqref{prod:uu'*uu'}} \notag \\
\pgd (N^2) &=&   24uu_0^\bu+8uu_0'+16uu_1^\bu+16uu_1' \label{pgd:N2}
\end{eqnarray}

The dmt-string $N^3$ is representable as $\hat D_2 \,\ov\odot_p\,\hat D_2$.  Therefore,\begin{eqnarray}  
\pgd (N^3) &=& \pgd(\hat D_2 \,\ov\odot_p\, \hat D_2) \notag \\
&=&  (2uu_0^\bu+2uu_0')\,\ov\odot_p\, (2uu_0^\bu+2uu_0') \notag \\
&=& 16uu_1^\bu + 16uu_1' + 16 uu_1' + 8uu_0^\bu +8uu_0' \notag \\
&& \qquad \hbox{by Prods.\,\eqref{prod:uu*uu}, \eqref{prod:uu*uu'}, \eqref{prod:uu'*uu}, \eqref{prod:uu'*uu'}} \notag \\
\pgd (N^3) &=& 8uu_0^\bu+8uu_0'+16uu_1^\bu+32uu_1' \label{pgd:N3}
\end{eqnarray}


\medskip
\section{\large{Amalgamating DMT-Strings} \label{sec:3to1}}  

A parallel operation on two dmt-strings yields a series-parallel graph whose source and target roots are both 2-valent.  Figure \ref{fig:uu*uu->dd} illustrates the productions used to  transform the partials of the two antecedent dmt-strings into the partials of the consequent series-parallel graph.  We define two partials for the case where $\firstterminal$ and $\secondterminal$ are both bivalent and a single walk is twice incident at each:
\begin{eqnarray*}
ss_i^\bu(G,\firstterminal,\secondterminal) &=&  \text{the number of embeddings } G\to S_i \text{ such that} \\ [-3pt]
&& \hbox{the fb-walks twice incident at $\firstterminal$ and $\secondterminal$ are different}; \\
ss_i'(G,\firstterminal,\secondterminal) &=& \hbox{the number of embeddings $G\to S_i$ such that} \\ [-3pt]
&& \hbox{the same fb-walk is twice incident at $\firstterminal$ and $\secondterminal$}. \\[-9pt]
\noalign{\noindent We also define \vskip3pt}  
dd_i''(G,\firstterminal,\secondterminal) &=& \hbox{the number of embeddings $G\to S_i$ such that} \\ [-3pt]
&& \hbox{the same two fb-walks are twice incident at $\firstterminal$ and $\secondterminal$}. 
\end{eqnarray*}

\begin{figure} [h!]  
\centering 
    \includegraphics[width=4in]{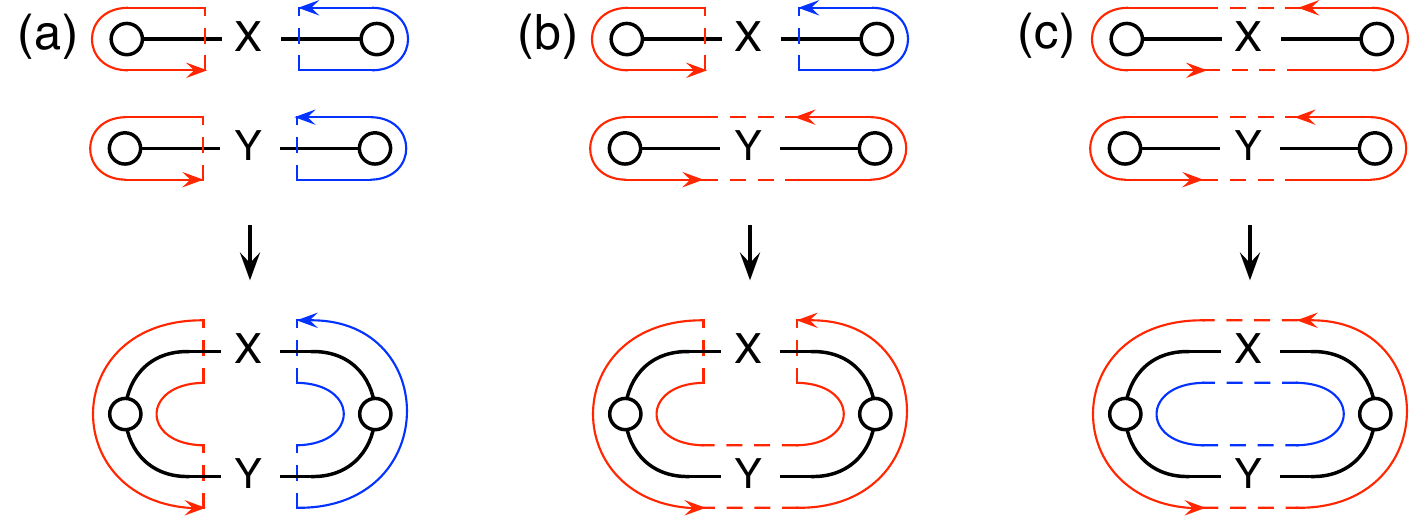}  
\caption{Parallel operations on two dmt-strings.}
\label{fig:uu*uu->dd}
\end{figure}  
\medskip

\begin{eqnarray} 
\noalign{\vskip-24pt}
uu_i^\bu \odot_p uu_j^\bu &\lra& ss_{i+j+1}^\bu \label{prod:uupuu} \qquad\hbox{(Figure \ref{fig:uu*uu->dd}(a))}   \\
uu_i^\bu \odot_p uu_j' &\lra& ss_{i+j+1}'\label{prod:uupuu'} \qquad\hbox{(Figure \ref{fig:uu*uu->dd}(b))}   \\
uu_i' \odot_p uu_j^\bu &\lra& ss_{i+j+1}'\label{prod:uu'puu} \qquad\hbox{(use Figure \ref{fig:uu*uu->dd}(b))}    \\
uu_i' \odot_p uu_j' &\lra& dd_{i+j}'' \label{prod:uu'puu'} \quad\qquad\hbox{(Figure \ref{fig:uu*uu->dd}(c))}     
\end{eqnarray}

Once again, the genus of each resultant embedding is calculated from its  Euler characteristic.  Continuing with our example of the graph from Figure 5.5, these productions enable us, in turn, to calculate \hbox{$\pgd(N^1\odot_p N^2)$}.

\begin{eqnarray}
\noalign{\vskip-12pt}
\qquad \pgd(N^1\odot_p N^2) &\!\!=\!\!& (12uu_0^\bu+4uu_0') \notag \\
&&\quad \odot_p (24uu_0^\bu+8uu_0'+16uu_1^\bu+16uu_1') \notag \\
&\!\!=\!\!& 12uu_0^\bu\odot_p (24uu_0^\bu+8uu_0'+16uu_1^\bu+16uu_1') \notag \\
&&\quad+~ 4uu_0'\odot_p (24uu_0^\bu+8uu_0'+16uu_1^\bu+16uu_1') \notag \\
&\!\!=\!\!& 288ss_1^{\bu}+ 96ss_1'+192ss_2^\bu+192ss_2' \notag \\
&&\quad +96ss_1'+32dd_0''+64ss_2'+64dd_1'' \notag \\
\;\; \pgd(N^1\odot_p N^2) &\!\!=\!\!& 32dd_0'' + 64dd_1'' + 288ss_1^\bu +192ss_2^\bu \label{pgd:N1-N2}\\
&&\quad +192ss_1'  +256ss_2' \notag 
\end{eqnarray}
\smallskip

The following six productions will enable us to complete our calculation of the genus distribution of the graph of Figure \ref{fig:runex}, starting from the pgd \eqref{pgd:N1-N2} for $N^1\odot_p N^2$ and the pgd \eqref{pgd:N3} for $N^3$.   We observe that the result of applying these productions is a genus distribution, rather than a partitioned genus distribution.  Accordingly, the consequents are of the form $g_i$ rather than subscripted partials, thereby indicating only that the resulting embedding surface is $S_i$.   
\smallskip

\begin{eqnarray} 
\noalign{\vskip-12pt}
dd_i'' \odot_p uu_j^\bu &\lra& 4g_{i+j+1} \label{prod:dd''0uu} \qquad\qquad\hbox{(Figure \ref{fig:dd''0uu})} \label{prod:dd''uu} \\
dd_i'' \odot_p uu_j' &\lra& 2g_{i+j} + 2g_{i+j+1} \label{prod:dd''0uu'} \quad\hbox{(Figure \ref{fig:dd''0uu'})} \\
ss_i^\bu \odot_p uu_j^\bu &\lra& 4g_{i+j+1} \label{prod:ss0uu} \qquad\qquad\hbox{(Figure \ref{fig:ss0uu})} \\
ss_i^\bu \odot_p uu_j' &\lra& 4g_{i+j+1} \label{prod:ss0uu'} \qquad\qquad\hbox{(Figure \ref{fig:ss0uu'})} \\
ss_i'\! \odot_p uu_j^\bu &\lra& 4g_{i+j+1} \label{prod:ss'0uu} \qquad\qquad\hbox{(Figure \ref{fig:ss'0uu})} \\
ss_i'\! \odot_p uu_j' &\lra& 4g_{i+j} \label{prod:ss'0uu'} \quad\qquad\qquad\hbox{(Figure \ref{fig:ss'0uu'})}  \label{prod:ss1uu'}  
\end{eqnarray}
\smallskip

Figures \ref{fig:dd''0uu}, \ref{fig:dd''0uu'}, \ref{fig:ss0uu}, \ref{fig:ss0uu'}, \ref{fig:ss'0uu}, and \ref{fig:ss'0uu'} illustrate the six productions \eqref{prod:dd''uu}, \eqref{prod:dd''0uu'}, \eqref{prod:ss0uu}, \eqref{prod:ss0uu'}, \eqref{prod:ss'0uu}, and \eqref{prod:ss'0uu'}, respectively, that collectively provide an algebraic representation of the parallel operation. 


\clearpage 
  
\begin{figure} [h!]  
\centering 
    \includegraphics[width=4.6in]{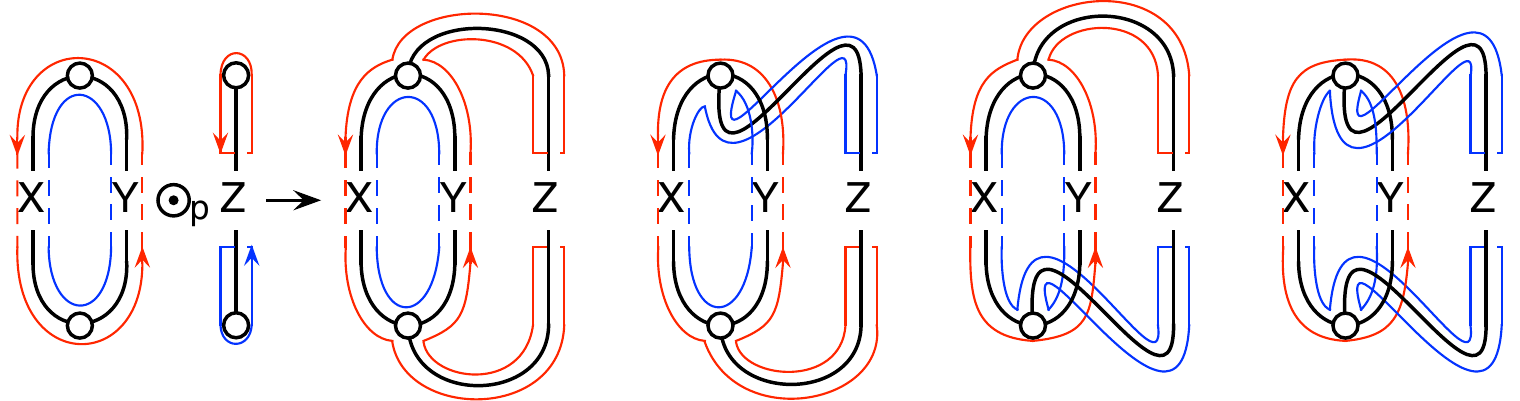}  
\caption{$dd_i''\odot_p uu_j^\bu \lra 4g_{i+j+1}$.}
  \vskip6pt 
\label{fig:dd''0uu}
    \includegraphics[width=4.5in]{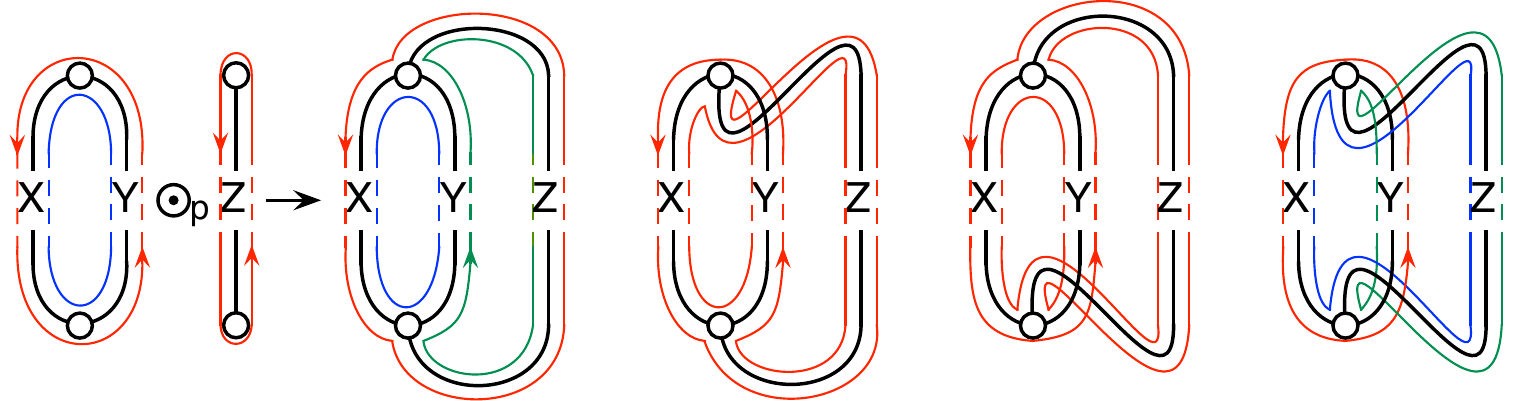}  \vskip-6pt
\caption{$dd_i''\odot_p uu_j' \lra 2g_{i+j}+2g_{i+j+1}$.}
  \vskip6pt 
\label{fig:dd''0uu'}
    \includegraphics[width=4.5in]{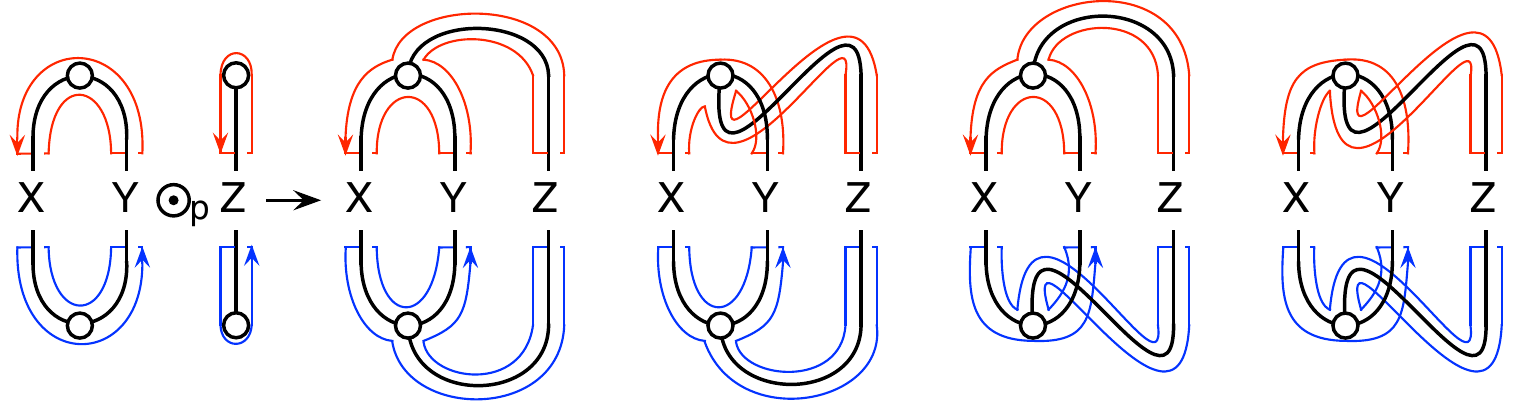}  \vskip-6pt
\caption{$ss_i^\bu\odot_p uu_j^\bu \lra 4g_{i+j+1}$.}
  \vskip6pt 
\label{fig:ss0uu}
    \includegraphics[width=4.5in]{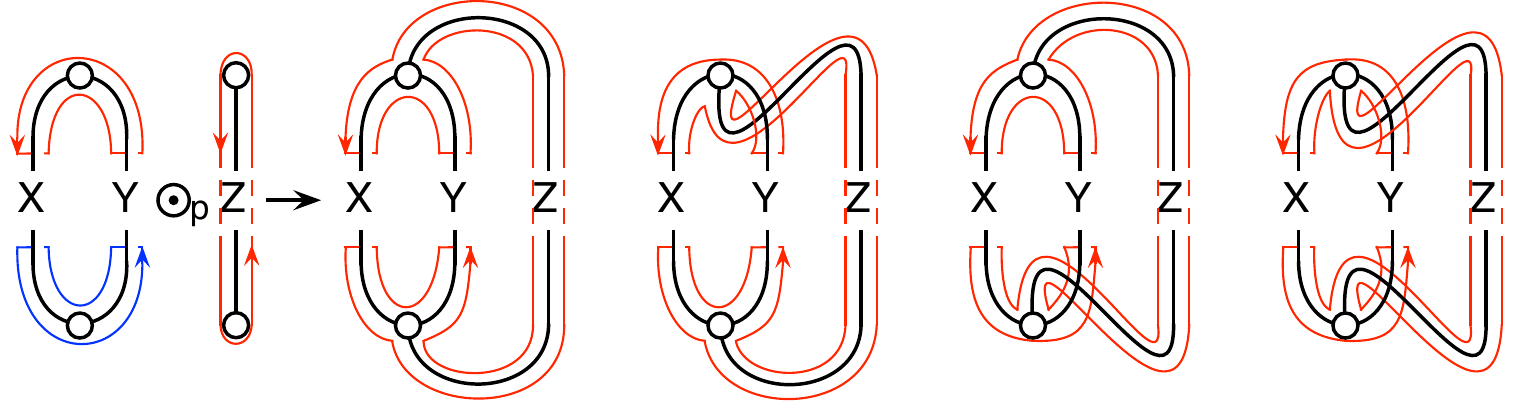}  \vskip-6pt
\caption{$ss_i^\bu\odot_p uu_j' \lra 4g_{i+j+1}$.}
  \vskip6pt 
\label{fig:ss0uu'}
    \includegraphics[width=4.5in]{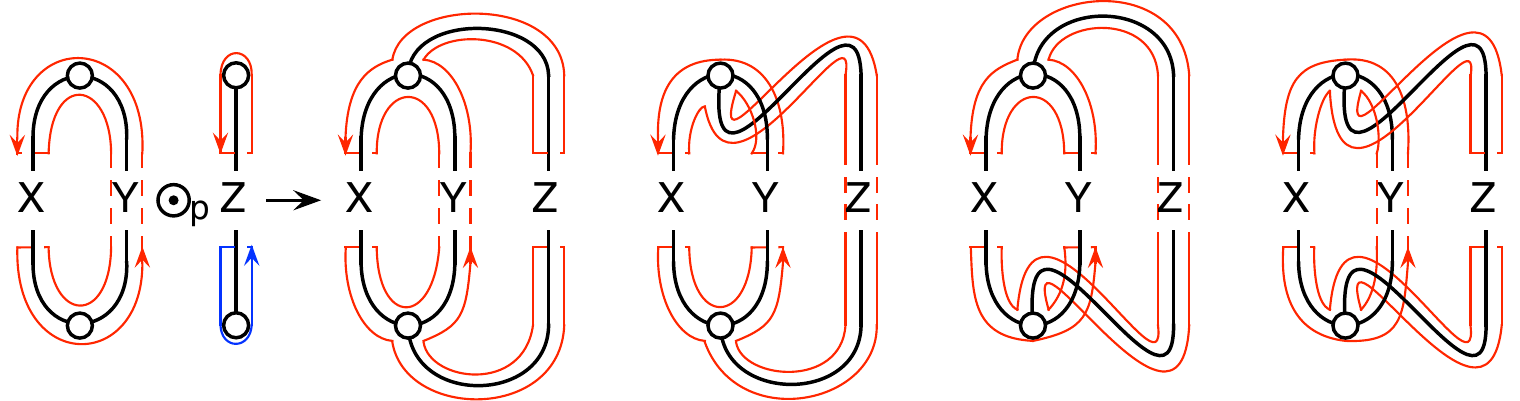}  \vskip-6pt
\caption{$ss_i'\odot_p uu_j \lra 4g_{i+j+1}$.}
\label{fig:ss'0uu}
\vskip-12pt
\end{figure}

\newpage

\begin{figure} [h!]  
\centering 
    \includegraphics[width=4.5in]{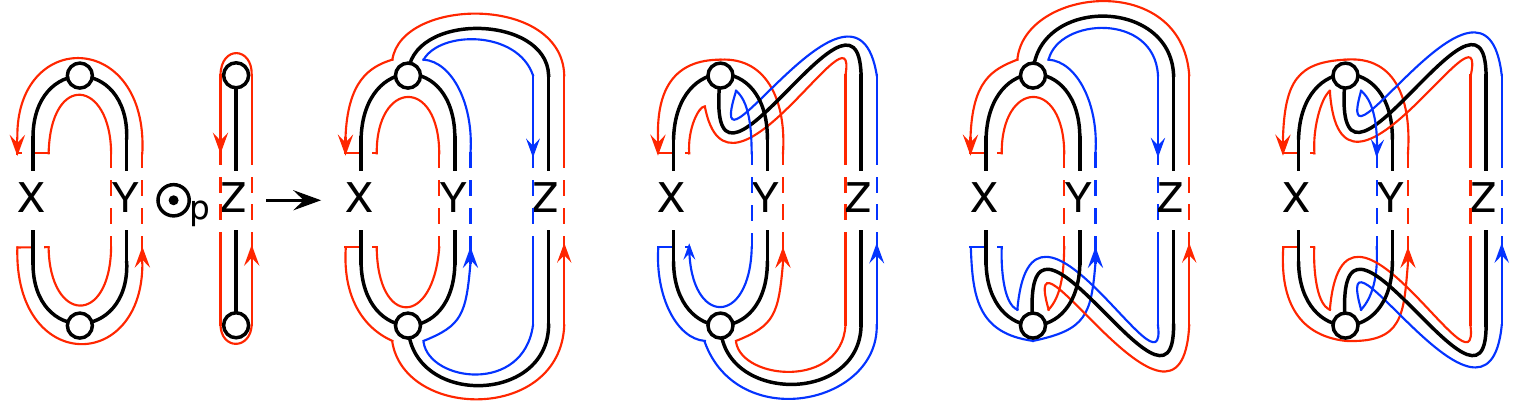}  \vskip-6pt
\caption{$ss_i'\odot_p uu_j' \lra 4g_{i+j}$.}
\label{fig:ss'0uu'}
\end{figure}  \vskip-12pt 

\enlargethispage{12pt}

\begin{eqnarray*}
\noalign{\vskip-18pt}
\pgd(N^1\odot_p N^2) &\!=\!& 32dd_0'' + 64dd_1'' + 288ss_1^\bu +192ss_2^\bu  +192ss_1'  +256ss_2'  \\
\pgd (N^3) &=& 8uu_0^\bu+8uu_0'+16uu_1^\bu+32uu_1' 
\end{eqnarray*}

{\allowdisplaybreaks
\begin{eqnarray*}  
\noalign{\vskip-12pt}
32dd_0''\odot_p 8uu_0^\bu &=& 1024g_1 \\
32dd_0''\odot_p 8uu_0' &=& 512g_0+512g_1 \\
32dd_0''\odot_p 16uu_1^\bu &=& 2048g_2 \\
32dd_0''\odot_p 32uu_1' &=& 2048g_1 + 2048g_2\\
64dd_1''\odot_p 8uu_0^\bu &=& 2048g_2 \\
64dd_1''\odot_p 8uu_0' &=& 1024g_1+1024g_2 \\
64dd_1''\odot_p 16uu_1^\bu &=& 4096g_3 \\
64dd_1''\odot_p 32uu_1' &=& 4096g_2 + 4096g_3\\
288ss_1^\bu\odot_p 8uu_0^\bu &=& 9216g_2 \\
288ss_1^\bu\odot_p 8uu_0' &=& 9216g_2 \\
288ss_1^\bu\odot_p 16uu_1^\bu &=& 18432g_3 \\
288ss_1^\bu\odot_p 32uu_1' &=& 36864g_3\\
192ss_2^\bu\odot_p 8uu_0^\bu &=& 6144g_3 \\
192ss_2^\bu\odot_p 8uu_0' &=& 6144g_3 \\
192ss_2^\bu\odot_p 16uu_1^\bu &=& 12288g_4 \\
192ss_2^\bu\odot_p 32uu_1' &=& 24576g_4\\
192ss_1'\odot_p 8uu_0^\bu &=& 6144g_2 \\
192ss_1'\odot_p 8uu_0' &=& 6144g_1 \\
192ss_1'\odot_p 16uu_1^\bu &=& 12288g_3 \\
192ss_1'\odot_p 32uu_1' &=& 24576g_2\\
256ss_2'\odot_p 8uu_0^\bu &=& 8192g_3 \\
256ss_2'\odot_p 8uu_0' &=& 8192g_2 \\
256ss_2'\odot_p 16uu_1^\bu &=& 16384g_4 \\
256ss_2'\odot_p 32uu_1' &=& 32768g_3 \\
\end{eqnarray*}
}

By summation of the right-hand sides of the equations above, we obtain the genus distribution
$$\begin{array}{c|rrrrr}
{\rm genus} & 0\hfil & 1\hfil & 2\hfil & 3\hfil & 4\hfil \\
\hline
{\rm \# embeddings} & 512 & 10752 & 68608 & 129024 & 53248 
\end{array} $$

This calculation has been confirmed by a computer program based on the Heffter-Edmonds algorithm. 
\smallskip

\begin{thm}
The time required by Algorithm \ref{algo:biconn} is at most quadratic in the number of vertices of the graph $G$ supplied as input.
\end{thm}

\begin{proof} 
Choosing an edge and its endpoints $s$ and $t$, as required by Step \eqref{algo:Epp}, takes constant time. 
The time needed for Step \eqref{algo:dmt-strings}, which is achieved by partitioning the edge-set of the given graph $G$ into the edge-sets of the  three dmt-strings, is linear in the number of vertices of $G$, using depth-first search.  

When two subgraphs are amalgamated during Step \eqref{algo:pgd-dmt}, the contribution to the partials of the merged graph corresponding to a pair of nonzero-valued subscripted partials,  one in the pgd of the first amalgamand and the other in the pgd of the second amalgamand, is calculated by the application of one of the Productions \eqref{prod:uu*uu}, \ldots, \eqref{prod:uu'-suu'}.  Since the number of nonzero-valued subscripted partials for any graph is linear in the number of vertices, the time to calculate the pgd of the merged graph is proportional to the product of the numbers of vertices in the two amalgamands.  Suppose that the numbers of vertices of the fragments of an dmt-string are $x_1, x_2, \ldots, x_p$.  Since the vertices in two different fragments will be merged into a combined fragment only once during the reassembly of the dmt-string, the number of applications of productions during the entire reassembly is at most 
$$\sum_{i\ne j}x_ix_j$$
However,
$$\sum_{i\ne j}x_ix_j ~<~ (x_1+x_2+\cdots+x_p)^2$$
from which we infer that the total time is at most quadratic in the number of vertices of the dmt-string. 

To see that Step \eqref{algo:pgd-2dmt} can be done in quadratic time, first observe that the number of nonzero partials in $\pgd(N^1)$ and $\pgd(N^2)$ is linear in the maximum genus of $N^1$ and $N^2$, since partials can be nonzero only for genera with genus less than or equal to maximum genus and 
for every  genus in this range there is only a constant number of different types of partials. Furthermore, it is well known that the maximum genus of any graph $G$ is bounded from above by $\beta(G)$ (see \cite{GrTu87}), the cycle rank of $G$, which in turn is linear in the number of vertices for any graph with maximum degree $3$. Accordingly, there is only a linear number of nonzero partials for both $N^1$ and $N^2$, yielding at most a quadratic number of combinations needed to calculate $\pgd(N^1\odot_p N^2)$. Finally, each combination of two partials can be computed in constant time using  appropriate choices from Productions \eqref{prod:uu*uu}, \ldots, \eqref{prod:uu'*uu'}, which implies that Step \eqref{algo:pgd-2dmt} can indeed be done in quadratic time in the number of vertices.

The argument that Step \eqref{algo:pgd-3dmt} uses only quadratic time is similar to that for Step \eqref{algo:pgd-2dmt} --- the number of nonzero partials is again linear in the number of vertices; and each computation, this time involving one of the Productions \eqref{prod:dd''uu}, \ldots, \eqref{prod:ss1uu'}, can be done in constant time.
\end{proof}

\bigskip
\section{\large{Extending to All Graphs of Treewidth $\le2$ and Maximum Degree $\le3$} \label{sec:TW2deg3}}  

The \mdef{bar-amalgamation} of two disjoint rooted graphs $(G,u)$ and $(H,v)$ is the result of running a new edge (the ``bar") between $u$ \hbox{and $v$}.  

\begin{prop}
The genus distribution of the bar-amalgamation of graphs $(G,u)$ and $(H,v)$ is the constant multiple of the convolution of the genus distributions of $G$ and $H$.  The constant factor is the product of the degree of $u$ in $G$ and the degree of $v$ in $H$. 
\end{prop}
\begin{proof} \vsb 
This is Theorem 5 of \cite{GrFu87}.
\end{proof}

To extend Algorithm \ref{algo:biconn} to any graph $G$ of treewidth \hbox{at most 2} and maximum degree 3, we infer from Proposition \ref{prop:TW2-SP} that each biconnected component of $G$  either is isomorphic to $K_2$ or is a homeomorphic copy of a 3-regular biconnected graph.  Moreover, each of the latter kind of biconnected components meets only the $K_2$-type components.  Thus, the graph $G$ is an iterated bar-amalgamation of biconnected series-parallel graphs.  Accordingly, to calculate its genus distribution, we calculate the genus distributions of its biconnected components, take corresponding convolutions, and multiply by scalars corresponding to degrees of vertices at the ends of the bars.

\bigskip
\section{\large{Conclusions}}  

Starting with a quadratic-time algorithm for calculating the genus distribution of all cubic biconnected series-parallel graphs, we have constructed a quadratic-time algorithm for the genus distribution of any graph of treewidth at most 2 and maximum degree at most 3.  The advantage of this case-specific algorithm over the more general algorithm of \cite{Gr12} is the ease with the case-specific algorithm can be used to obtain numerical results for specific graphs.

\bigskip
\section*{{Acknowledgements}}   

Thanks to Maria Chudnovsky for suggesting series-parallel graphs as an interesting family of low-treewidth graphs for a genus distribution calculation.  

The second author was partly supported by Nad\'acia Tatra Banky grant  11sds071, SAIA NSP, grant APVV-0223-10,  the grant APVV-ESF-EC-0009-10 within the EUROCORES Programme EUROGIGA (project GReGAS) of the European Science Foundation, and Ministry of Education, Youth, and Sport project No.~CZ.1.07/2.3.00/30.0009 -- Employment of Newly Graduated Doctors of Science for Scientific Excellence.  The research was done while the second author was a Ph.D.~student at the Department of Computer Science, Comenius University, Bratislava, Slovakia, and was 
visiting the Department of Computer Science at Columbia University.  He would like to thank his host Prof. J.~L.~Gross and the department for the hospitality.

\bigskip

\end{document}